\documentclass[journal,twocolumn]{IEEEtran}

\IEEEoverridecommandlockouts
\usepackage{cite}
\usepackage{amsmath,amssymb,amsfonts}
\usepackage{multirow}
\usepackage{graphicx}
\usepackage{subfig}
\usepackage{textcomp}
\usepackage{xcolor}
\usepackage[utf8]{inputenc}
\usepackage[english]{babel}
\usepackage{fancyhdr}
\usepackage[font=footnotesize,skip=0pt]{caption}
\usepackage{array}
\usepackage{tabularx}
\usepackage[ruled,vlined]{algorithm2e}
\usepackage{algpseudocode}

\usepackage{amsthm}

\newtheorem{lemma}{Lemma}

\newtheorem{remark}{Remark}
\newtheorem{corollary}{Corollary}

\def\BibTeX{{\rm B\kern-.05em{\sc i\kern-.025em b}\kern-.08em
    T\kern-.1667em\lower.7ex\hbox{E}\kern-.125emX}}
\begin{document}

\title{Interference and Coverage Analysis in Coexisting RF and Dense TeraHertz Wireless Networks\\
%{\footnotesize \textsuperscript{*}Note: Sub-titles are not captured in Xplore and
%should not be used}
%\thanks{Identify applicable funding agency here. If none, delete this.}
}
\author{\IEEEauthorblockN{Javad Sayehvand}
% \IEEEauthorblockA{\textit{dept. name of organization (of Aff.)} \\
% \textit{name of organization (of Aff.)}\\
% City, Country \\
% email address}
and
\IEEEauthorblockN{Hina Tabassum, {\em Senior Member IEEE}}
% \IEEEauthorblockA{\textit{dept. name of organization (of Aff.)} \\
% \textit{name of organization (of Aff.)}\\
% City, Country \\
% email address}
%\and
%\IEEEauthorblockN{ Hina Tabassum}
% \IEEEauthorblockA{\textit{dept. name of organization (of Aff.)} \\
% \textit{name of organization (of Aff.)}\\
% City, Country \\
% email address}
%\and
%\IEEEauthorblockN{4\textsuperscript{Th} Given Name Surname}
%\IEEEauthorblockA{\textit{dept. name of organization (of Aff.)} \\
%\textit{name of organization (of Aff.)}\\
%City, Country \\
%email address}
%\and
%\IEEEauthorblockN{5\textsuperscript{Th} Given Name Surname}
%\IEEEauthorblockA{\textit{dept. name of organization (of Aff.)} \\
%\textit{name of organization (of Aff.)}\\
%City, Country \\
%email address}
%\and
%\IEEEauthorblockN{6\textsuperscript{Th} Given Name Surname}
%\IEEEauthorblockA{\textit{dept. name of organization (of Aff.)} \\
%\textit{name of organization (of Aff.)}\\
%City, Country \\
%email address}
\vspace{-10mm}
\thanks{~J. Sayehvand and H.~Tabassum are with the  York University, Canada
(e-mail: jsayehvand@cse.yorku.ca and hina@eecs.yorku.ca). This work is supported by the Discovery
Grant from the Natural Sciences and Engineering Research Council of Canada.}
}

\raggedbottom

\maketitle

\begin{abstract}

This paper develops a stochastic geometry framework to characterize the  statistics of the downlink interference and coverage probability of a typical user in a coexisting terahertz (THz) and radio
frequency (RF) network. We first characterize the exact Laplace Transform (LT) of the aggregate interference and coverage probability of a user in a {THz-only} network. Then, for a {coexisting RF/THz  network}, we derive the coverage probability of a typical user considering biased received signal power association (BRSP). The framework can be customized to
capture the performance of a typical user in various network configurations such as  THz-only, opportunistic RF/THz, and hybrid RF/THz.
In addition, asymptotic approximations are presented for scenarios where the intensity of THz BSs becomes large or molecular absorption coefficient in THz approaches to zero. Numerical results demonstrate the accuracy of the derived expressions and extract insights related to the significance of the BRSP association compared to the
conventional reference signal received power (RSRP) association in the coexisting   network. 
\end{abstract}

\begin{IEEEkeywords}
Terahertz (THz), interference, coverage probability, stochastic geometry, Poisson point process (PPP).
\end{IEEEkeywords}

%\section{Introduction}
%This document is a model and instructions for \LaTeX.
%Please observe the conference page limits. 
\vspace{-3mm}
\section{Introduction}

% The estimated growth in the global mobile traffic volume has increased from 7.462 EB/month in 2010 to 5016 EB/month in 2030 \cite{yastrebova2018future}. Thus, 
{The potential of using higher frequency spectrum such as terahertz (THz) in the sixth generation (6G) wireless networks is evident~\cite{boulogeorgos2019analytical}.}
% THz frequencies (0.1-10 THz)  can potentially enable sophisticated applications (e.g., virtual reality and augmented reality, vehicular networks,  industry 4.0, etc.) that require agile, reliable, and almost zero latency transmissions.   
%{In this regard, some of the leading efforts include "6Genesis Flagship Program (6GFP)" \cite{katz20186genesis} and TERRANOVA sponsored by European Commission’s Horizon 2020 program. 
% %The Semiconductor Research Corp. sponsored research on the convergence of THz communications and sensing technologies for future cellular infrastructure. T
% The U.S. Federal Communications Commission (FCC) recently opened the THz band for research purposes in 2019.} 
THz frequencies offer ample spectrum, multi Gigabit-per-second (Gbps) data rates, and highly secure transmissions. Nonetheless, compared to conventional radio frequency (RF),  \textcolor{black}{THz transmissions incur high propagation loss mainly due to molecular absorption  resulting from water vapors or oxygen molecules}, thus significantly limiting the communication distance. THz spectrum is thus complementary to conventional RF spectrum.
\textcolor{black}{Recent advancements have made it possible to mount THz transceivers on smart phones, e.g., Fujitsu who introduced a compact 300 GHz transceiver capable of 20 Gbps data stream.  \cite{nakasha2017compact}.} 
%As an example, for a 10 m at 300 GHz  is 102 dB, which is almost twice of its microwave counterpart at 2.1 GHz []. \textcolor{black}{Subsequently, the THz transmissions will coexist and complement  the conventional radio frequency (RF) transmissions.}   
%Fortunately, the possibility of packing hundreds of small antennas even in hand-held device can alleviate the propagation loss to some extent, yet data transmission at THz frequencies over tens of meters seems to remain challenging. As such, the prominent feature of THz data transmission; that is, shot-range as well as super-high data rate makes THz links another enabler of networks beyond 5G to fulfill the expected use cases.  
% Nonetheless, the performance of the THz transmissions can be impeded by its susceptibility to line-of-sight (LOS) blockages, molecular absorption, and communication range.  Thus, the coverage zones are limited compared to RF (sub-6GHz). % As such, THz transmissions will coexist with the conventional radio frequency (RF) spectrum.  

To date, most of the research works considered analyzing the performance of a given THz transmission link \cite{boulogeorgos2019analytical} or  THz-only network \cite{kokkoniemi2017stochastic,7820226, chaccour2019reliability}. For instance, the authors in \cite{boulogeorgos2019analytical} derived a closed-form expression of the outage probability and ergodic capacity considering a THz wireless fiber extender system (i.e., a single transmission link) with ideal and non-ideal RF front-end. Using tools from stochastic geometry and  considering interference limited regime,  the authors derived the mean interference in a THz-only network \cite{kokkoniemi2017stochastic}. However, the closed-form expression of the mean interference was neither applicable for a general case, nor the expression was applied to the outage analysis. Instead, the authors approximate the distribution of the interference with log-logistic distribution to overcome the intractable outage calculation. Nevertheless, as the authors  mentioned, the use of log-logistic approximation might not be accurate in all scenarios. \textcolor{black}{Similarly, \cite{7820226} considered Taylor expansion and calculated the approximations for mean and variance of signal-to interference-plus-noise ratio (SINR).} 
%\textcolor{black}{Moreover, the general free space path-loss model provided in the paper is different from other well-known papers in this field, bring a shadow of doubt that there is a typo in the paper if they have not performed wrong analysis.}   
%It has been argued that the molecular noise is very weak if we see the temporal distribution of the noise or if we consider a case when the absorbed energy gets transformed into heat. 
In \cite{chaccour2019reliability}, the authors analyzed the reliability and end-to-end latency  considering a  THz-only network with finite number of BSs. The interference  was approximated with a normal distribution.
% In \cite{yao2019interference}, 
% the interference and outage probability in a THz-only network was investigated. The interference was approximated with mean interference. 
The authors in \cite{ntontin2016toward}  derived the approximate coverage probability  in a single-tier network, where BSs can use either RF or THz. 

To our best knowledge, none of the aforementioned research works presented a comprehensive analytic framework to characterize the exact interference statistics  and coverage probability of users in a THz-only network or a coexisting two-tier RF and dense THz network
% (i.e., where the intensity of THz BSs (TBSs) can be significantly high compared to the conventional RF small base stations (SBSs)). 
%Indoor T-waves communication is factored in in terms of achievable as well as coverage probability in \cite{moldovan2017coverage}. The authors suggested to apply single single frequency network (SNF) to overcome the limited output power of T-waves transmitters. In addition, the performance of multi-user T-waves networks is evaluated with regard to the SNR per transmitted symbol at the output of each equalizer. To improve the performance of the networks, the authors have formulated an optimization problem to allocate power in an decent manner when frequency division scheduling (FDS) is taken into account, and they introduced an algorithm to solve it. . However, this paper did not deal with calculating the exact outage probability and achievable data rate.     

%\cite{jornet2011channel} tried to characterize the T-waves propagation model for wireless communication, and to this aim, they have considered experimental and statistical approaches total path-loss and molecular absorption noise.

Using stochastic geometry, this paper  characterizes the statistics of the downlink interference and rate coverage probability of a typical user in a coexisting RF/THz network. \textcolor{black}{The proposed framework can be customized for various network configurations, including (i) \textit{THz-only} network where only TBSs exist and users associate to their nearest BS,  (ii) \textit{opportunistic RF/THz} network where a user associates to the BS with maximum biased received signal power (BRSP)\footnote{\textcolor{black}{BRSP-based association is considered  by 3GPP in Release 10, where the users’ power received from small base stations (SBSs) has been artificially increased by adding a bias in order to avoid under-utilization of SBSs \cite{liu2016user}.}}, and (iii)\textit{ Hybrid} network where a user associates to both nearest RF and TBSs.} We first characterize the exact Laplace Transform (LT) of the aggregate interference and coverage probability of a user in a \textit{THz-only} network. Then, we derive the coverage probability of a typical user in  a \textit{coexisting  network}. Asymptotic approximations are presented for large intensity of TBSs or small molecular absorption coefficients. Numerical results show the significance of BRSP over
conventional reference signal received power (RSRP) association in a coexisting network and validate the derived expressions.

% Users then associate  to the chosen BS  and the BS serves them in orthogonal channels or time slots. In this letter, devices are allowed to associate with only one BS at a time in {\em RF-only, THz-only, Opportunistic RF/THz}, whereas hybrid scheme allows a user device to associate with both RF SBSs and TBSs. All users associated to a BS gets served on orthogonal channels or time slots.
% Finally, the coverage probability of co-existing RF/THz is explored. It is worth-mentioning that, this investigation is not simply combining the analytical result for coverage probability of RF and THz links, as the two tired network changes the distance distribution of a user form its tagged. Thus, the main contribution here is extracting the distance distribution based on which the co-exiting coverage probability is calculated.  

%The rest of this paper is organized as follows. The system model and its assumptions are explained in Section II. In Section III, the SG based model for interference is derived, and the closed-form formula for the Laplace transform of interference is calculated. In addition, in this section, the Laplace transform of interference is leveraged to calculate the coverage probability. The simulation and numerical results are provided and compared in Section IV to confirm the validity of the achieved derivations. Finally, in section V, the conclusion is provided.

\section{System model and Assumptions}
We consider a two-tier downlink network composed  of RF SBSs and TBSs as well as users' devices.
The locations of the conventional RF SBSs and TBSs are modeled as a two-dimensional (2D) homogeneous Poisson point processes (PPP)  $\mathbf{\Phi}_{R}$ and $\mathbf{\Phi}_{T}$ with intensities  $\lambda_{R}$ and $\lambda_{T}$, respectively. The locations of the users follow independent homogeneous PPP  $\mathbf{\Phi}_{u}$ with intensity $\lambda_{u}$.  \textcolor{black}{Each user measures the channel quality from each BS and then associate  to the chosen BS according to a predefined association mechanism.  The BSs serve associated users in orthogonal time slots or channels.}  We  consider the performance of a typical user who is located at the origin.

\subsubsection{ RF Channel and SINR Model}  The RF channel experiences both the channel fading  and path-loss. Thus, the received signal power at the typical user can be modeled as {$h(\rho)= \gamma_{R} \rho^{-\alpha}\chi$}, where  $\gamma_{R} = \frac{c^2}{\left(4\pi f_{R}\right)^2} $, $\chi$ is the exponentially distributed channel power with unit mean from the tagged SBS, $\alpha$ is the path-loss exponent, $\rho$ is the distance of the typical user to the serving SBS, $f_{R}$ is the RF carrier frequency in GHz, and $c = 3 \times 10^8$ m/s is the speed of light. We consider SBSs equipped with omni-directional antennas. {Therefore, the SINR of a typical user  can be {modeled as:}
\begin{equation}
\mathrm{SINR}_{R}= \frac{P_{R}\gamma_{R}\rho_0^{-\alpha}\chi_0}{N^R_{0}+I_{\mathrm{agg}}^{R}},
\end{equation}
where $\chi_0$ is the fading channel power of the typical user from the desired SBS, $P_{R}$ is the transmit power of the SBSs, $N^R_{0}$ is the thermal noise at the receiver, $I_{\mathrm{agg}}^{R} = \sum_{i\in \Phi_{R}\backslash 0}P_{R} \gamma_{R}\rho_{i}^{-\alpha}\chi_{i}$  is the aggregate interference at the typical user from the interfering SBSs, $\rho_{i}$ is the distance between the $i$-th interfering SBS and the typical user, and $\chi_{i}$ is the fading channel power from the $i$-th interfering SBS. 
} 

\subsubsection{THz Channel and SINR Model}  
% Since the molecular absorption loss is high in THz, the impact of multi-path fading and non-line-of-sight (NLoS) transmission is negligible within short distances \cite{moldovan2014and}.  
\textcolor{black}{Due to  high molecular absorption and the dense deployment, the  LoS transmissions are  dominant than the NLoS transmissions. Therefore,  following  \cite{kokkoniemi2017stochastic,mumtaz2017terahertz,chaccour2019reliability,chaccour2020can}}, we model the line-of-sight (LoS) channel power\footnote{\textcolor{black}{The consideration of NLoS with accurate reflection, scattering, and diffraction models deserves a separate study and has been left for future investigation.}} between users and TBSs  as 
\textcolor{black}{$h\left(r\right) =\frac{c^2}{\left(4\pi f_{T}\right)^2} \frac{\exp\left(-k_{a}\left(f\right)r\right)}{r^2},$} where $k_{a}\left(f\right)$
 is the molecular absorption coefficient, $r$ is the distance between the transmitter and  receiver, $f_{T}$ is the operating THz frequency, $c$ is the speed of light, and $G^{\mathrm{T}}_{\mathrm{tx}}\left(\theta\right)$ as well as $G^{\mathrm{T}}_{\mathrm{rx}}\left(\theta\right)$ are the directional transmitter and receiver antenna gains, respectively. The directional antennas are  modeled as \cite{di2015stochastic}:   
\begin{equation}
\label{eq:gain}
  G^{T}_{q} \left(\theta\right) =
    \begin{cases}
      G_{q}^{\left(\mathrm{max}\right)} & \mid \theta \mid \leq w_{q}\\ 
      G_{q}^{\left(\mathrm{min}\right)} & \mid \theta \mid > w_{q}
    \end{cases},  
\end{equation}
where $q\in \{\mathrm{tx,rx}\}$ , $\theta \in [-\pi,\pi)$
     is the angle of the boresight direction, $w_{q}$ is the main lobe beamwidth, $G_{q}^{\left(\mathrm{max}\right)}$ and $ G_{q}^{\left(\mathrm{min}\right)}$ are beamforming gains of the main and side lobes, respectively. The typical user and its desired TBS align such that their main lobes coincide through beam alignment techniques. \textcolor{black}{The alignment between the typical user and an interferer is defined as a random variable $D$, which can take values in} $\{G_{\mathrm{tx}}^{\left(\mathrm{max}\right)}G_{\mathrm{rx}}^{\left(\mathrm{max}\right)},G_{\mathrm{tx}}^{\left(\mathrm{max}\right)}G_{\mathrm{rx}}^{\left(\mathrm{min}\right)},G_{\mathrm{tx}}^{\left(\mathrm{min}\right)}G_{\mathrm{rx}}^{\left(\mathrm{max}\right)},G_{\mathrm{tx}}^{\left(\mathrm{min}\right)}G_{\mathrm{rx}}^{\left(\mathrm{min}\right)}\}$, and the corresponding probability for each case is $F_{\mathrm{tx}}F_{\mathrm{rx}}$, $F_{\mathrm{tx}}(1-F_{\mathrm{rx}})$, $(1-F_{\mathrm{tx}})F_{\mathrm{rx}}$, and $(1-F_{\mathrm{tx}})(1-F_{\mathrm{rx}})$, where $F_{\mathrm{tx}} = \frac{\theta_{\mathrm{tx}}}{2\pi}$ and $F_{\mathrm{rx}} = \frac{\theta_{\mathrm{rx}}}{2\pi}$, respectively. {Assuming that the main lobe typical user's receiver  is coinciding with that of its desired TBS, its SINR can be formulated as follows:}
\begin{equation}
\mathrm{SINR}_{T}= \frac{P_{\mathrm{T}}\textcolor{black}{G^{(\mathrm{max})}_{\mathrm{tx}}\left(\theta\right)G^{(\mathrm{max})}_{\mathrm{rx}}\left(\theta\right)}  \frac{c^2}{\left(4\pi f_{T}\right)^2} \frac{\exp\left(-k_{a}\left(f\right)r\right)}{r^2}}{N^T_{0}+I_{\mathrm{agg}}^{T}},
\end{equation}
where \textcolor{black}{$I^T_{\mathrm{agg}} =  \sum_{\mathrm{i} \in \Phi_{T} \backslash 0}P_{\mathrm{T}}D_{i} h\left(r_{i}\right)$} is the aggregate interference at the typical user by their maximum gain, $r_{i}$ is the distance of the typical user to the interfering TBSs. For brevity, we define $\gamma_{T} = G^{\mathrm{(max})}_{\mathrm{tx}}\left(\theta\right)G^{\mathrm{(max})}_{\mathrm{rx}}\left(\theta\right)\frac{c^2}{\left(4\pi f_{T}\right)^2}$. We assume that the interferers' main lobe coincides with the users' main lobe\footnote{\textcolor{black}{For simplicity, we consider negligible side lobe gains. However, the framework can be extended by averaging over variable $D$ and considering all four possible interference components in $I^{\mathrm{T}}_{\mathrm{agg}}$ with different antenna gains. 
These four interference variables are independent  and their LTs can be given using {\bf Lemma~1}. $\mathcal{L}_{I^T_{\mathrm{agg}}\mid r}(s)$  can thus be given as the product of their LTs. 
% implying that $I^T_{\mathrm{agg}} = \sum_{\mathrm{j,k} \in \psi }\sum_{\mathrm{i} \in \Phi_{T}\backslash 0}F_{i}F_{j}P_{\mathrm{T}}\gamma_T G^{\mathrm{T,j}}_{\mathrm{tx}}\left(\theta\right)G^{\mathrm{T},k}_{\mathrm{rx}}\left(\theta\right)h\left(r_{i}\right)$, where $\psi = \{\mathrm{max},\mathrm{min}\}$}.
}} with the probability of $F = F_{\mathrm{tx}}F_{\mathrm{rx}}$, and thus, $D = G^{(\mathrm{max})}_{\mathrm{tx}}\left(\theta\right)G^{(\mathrm{max})}_{\mathrm{rx}}\left(\theta\right)$. Also, $P_{T}$ is the transmit power of the TBSs, and $N^T_{0}$ denotes the thermal noise and the noise resulted from the molecular absorption which is considered as negligible in dense THz networks.

%  It is worth mentioning that we have ignored the effect of blockages for simplicity. 
%  \textcolor{cyan}{However, the same model in \cite{bai2014analysis} for blockages can be considered to capture this effect, and  fortunately since the probability of LoS in this model is a form of order polynomial as its power, the procedure of analysis will not be effected except for except for some changes in the constants in the final results.}

%  which depends on the frequency of operation and the length of the propagation link.
%  However, it cannot be ignored in THz band, and thus, the THz noise is:
% $
% N^T_{0}= N^T_{0} + N^T_{abs}\left(f^T, R\right) 
% $
% where $R$ is the length of the propagation link.
%%%%%%%%%%%%%%%%%%%%%%%%%%%%%%%%%%%%%%%%%%%%%%%%%%%%%%%%%%%%%%%%%%%%%%%%%%%%%%
\section{Coverage Probability in THz-Only Network}
In this section, we derive the LT of the aggregate interference and the coverage probability experienced  by a typical user in a THz-only network with nearest BS association.  
%By assuming interference limited case, the received $SINR$ will be $\frac{S}{I_{\mathrm{agg}}}$.
\begin{lemma}
\label{lm:Laplace_Trans}
Conditioned on the distance of a typical user from the serving TBS, the  LT of the aggregate interference,  at a typical device in THz network, can be derived  as follows:
\small
\begin{equation*}
\label{eq:Laplace_Trans}
\begin{split}
\mathcal{L}_{I^T_{\mathrm{agg}}\mid r}(s) = \exp\Biggl(2\pi\lambda_{T}\sum_{l=1}^{\infty}\frac{\left(-s \gamma_{T} F P_{\mathrm{T}}\right)^{l} \Gamma\left(2-2l,lk_{a}\left(f\right)r)\right)}{\left(l k_{a}\left(f\right)\right)^{2-2l}l!}
\Biggr).
\end{split}
\end{equation*}
\normalsize
\end{lemma}
\begin{proof}
See \textbf{Appendix A}.
\end{proof}
% We note that the infinite summation in \textbf{Lemma~1} can be truncated to first three terms without loss of significant accuracy. The comparisons  are presented in Fig.~1(a) and (b).
%As it is obvious from Table \ref{Tab1}, the value of interference Laplace transform for first-three-term and first-twenty-term approximation only differ about $0.000001$, and what is more, the difference between first-term and first-three-term approximation is $0.00215 \approx 0.002$. This means, dependent on the required precision, either of first-term and first-three-term approximation will perform sufficiently.      

We define the rate coverage probability as the probability with which a typical user achieves the desired target rate  $R_{\mathrm{th}}$. Subsequently, using  $R_{\mathrm{th}}=W_T\mathrm{log}_2(1+\mathrm{SINR})$ (where $W_T$ is the THz transmission bandwidth), the rate coverage probability can thus be given as follows: 
\begin{equation}
\label{eq:cov_prob_def}
\mathcal{P}_{T}= \mathrm{Pr}\left(\mathrm{SINR}_T>2^{\frac{R_{\mathrm{th}}}{W_T}}-1\right)  = \mathrm{Pr}\left(\mathrm{SINR}_T>\tau_T\right).
\end{equation}
Taking the desired signal power at the typical user $S(r) = P_{T}\gamma_{T}\frac{\exp\left(-k_{a}\left(f\right)r\right)}{r^2}$ and using the Gil-Pelaez inversion theorem \cite{gil1951note}, $\mathcal{P}_{T}$ can be derived as follows:
\begin{align}
\label{eq:Twave_cov_only}
&\mathcal{P}_{T} = \mathrm{Pr}\left(\frac{S(r)}{N^T_{0}+I^T_{\mathrm{agg}}}>\tau_T\right)
= \mathrm{Pr}\left(S(r)>\tau_TI^T_{\mathrm{agg}}+\tau_TN^T_{0}\right), 
\nonumber\\&{=} 
\mathbb{E}_{r}[\frac{1}{2}-\frac{1}{\pi}\int_{0}^{\infty}\frac{\mathrm{Im}[\phi_{\Omega|r}(\omega)e^{j\omega \tau_T N^T_0}]}{\omega}d\omega],
\nonumber\\&{=} 
\frac{1}{2}-\frac{1}{\pi}\int_{0}^{\infty}\frac{\mathrm{Im}[\phi_{\Omega}(\omega) e^{j \omega \tau_T N^T_0}]}{\omega}d\omega,
\end{align}
where $\mathrm{Im}(\cdot)$ is the imaginary part of $\phi_{\Omega}\left(\cdot\right)$, $\Omega =S(r) - \tau_TI^T_{\mathrm{agg}}$, and $\phi_{\Omega}\left(w\right) = \mathbb{E}[e^{-j\omega\Omega}]$ is the characteristic function (CF) of $\Omega$ given as follows: 
%\begin{enumerate}
%\item
%As $X$ and $I^T_{\mathrm{agg}}$ are dependent on $r$, it is more straightforward to compute the conditional CF, that is, $\phi_{\Omega \mid r}\left(\cdot\right)$. Thus, we should firstly derive the conditional Laplace transform and CF of the aggregate interference $(I_{\mathrm{agg}})$.   
%Derive CF $\phi_{\Omega}\left(\cdot\right)$ as follows:
%\small
\begin{equation*}
\label{eq:CF}
\phi_{\Omega}\left(\omega\right) = \mathbb{E}_{r}\left[\phi_{\Omega \mid r}\left(\omega\right)\right] = \\ \mathbb{E}_{r}[e^{-j\omega {S (r)}}\mathcal{L}_{I^T_{\mathrm{agg}}\mid r}\left(-j\omega \tau_T\right)],
\end{equation*}
where $\mathcal{L}_{I^T_{\mathrm{agg}}\mid r}$ is given in {\bf Lemma~1}. Gil-Pelaez inversion  is  applicable to the CF of any random variable and has been proved useful in a wide variety of wireless applications~\cite{tabassum2018coverage}. 
% It is also worth mentioning that Gil-Pelaez inversion theorem is a practical tool for calculating coverage probability and a large body of papers have used it without any report on the singularity and inefficiency \cite{tabassum2018coverage,hattab2018coverage,guan2014stochastic,7707400}
%\normalsize
%\item Now, use \eqref{eq:Twave_cov_only} to determine the coverage probability. 
%\end{enumerate}
\section{Coverage in Coexisting RF/THz Network}
In the coexisting  network, the user can either associate to a given SBS or TBS based on the maximum BRSP with a probability termed as {\em association probability}.

Given the received powers from TBSs and SBSs  as  $P_{r}^{\mathrm{THz}} = P_{{T}}\gamma_T\times \frac{\exp\left(-k_{a}\left(f\right)r\right)}{r^2}$ and $P_{r}^{\mathrm{RF}} = P_{{R}}\gamma_R \times\rho^{-\alpha}$, respectively, the probability of association to TBS can be derived as follows:
\begin{equation}
\label{eq:Asso_Probab_Proof}
\begin{split}
&\mathcal{P}_{A_{T}} = \mathbb{E}_{r}\left[\mathrm{Pr}\left[B_{T} P_{r}^{\mathrm{THz}}>P_{r}^{\mathrm{RF}}\right]\right],
\\ & = \mathbb{E}_{r}\left[\mathrm{Pr}\left[P_{{T}}B_{T}\gamma_{T}\frac{\exp\left(-k_{a}\left(f\right)r\right)}{r^2}>P_{{R}}\gamma_{R}\rho^{-\alpha}\right]\right],
% \\ & = \mathbb{E}_{r}\left[\mathrm{Pr}\left[\rho > \left(K r^2\exp\left(k_{a}\left(f\right)r\right)\right)^{\frac{1}{\alpha}}\right]\right],
\\ & \stackrel{(a)}{=} \mathbb{E}_{r}\left[\exp \left(-\pi\lambda_{R}\left(K r^2\exp\left(k_{a}\left(f\right)r\right)\right)^{\frac{2}{\alpha}}\right)\right],
\end{split}
\end{equation}
\normalsize
where $K=\frac{P_{R}\gamma_{R}}{B_{T} P_{T}\gamma_{T}}$, $B_{T}  >1$ is the bias value to encourage association with THz layer, $0 \leq B_T < 1$ encourages association to RF layer, \textcolor{black}{$B_T = 1$ yields conventional RSRP} and $\left(a\right)$ follows from the null property of PPP $\Phi_R$. This property implies that given a tier of RF SBSs with intensity $\lambda_{\mathrm{R}}$, the probability that no RF BSs are closer to typical user than the distance $z$ is $\mathbb{P}\left[\rho\geq z\right] = \exp \left(-\pi \lambda_{R}z^2\right)$.

Taking the expectation, the association probability of a typical user to TBSs can be derived as in the following lemma.
\begin{lemma}
\label{lm:Assocation_probability}
 Given that the  user associates with the layer that provides the maximum BRSP, the probability of association to the THz layer is given as follows:
\small
\begin{equation}
\label{eq:Assocation_probability_TWaves}
\mathcal{P}_{A_{T}} =\sum_{j=0}^{\infty} \frac{\left(-1\right)^j \Gamma\left[v_{j}\right] \delta_{T,j}}{\left(2\beta_{T}\right)^{\frac{v_{j}}{2}}j!} \exp{\left(\frac{-\eta_{j}}{8\beta_{T}}\right)} D_{-v_{j}}\left(\frac{-\eta_{j}}{\sqrt{2\beta_{T}}}\right),
\nonumber
\end{equation}
\normalsize
 where $\beta_{T} = \pi\lambda_{T}$, $\delta_{T,j} = 2\pi\lambda_{T}\left(\pi\lambda_{R} K^{\frac{2}{\alpha}}\right)^j$, $\nu_{j} = \frac{4j}{\alpha}+2$, $\eta_{j} = - \frac{2jk(f)}{\alpha}$, and  $D_{\nu}\left(z\right)$ is the parabolic cylinder function (PCF) (\cite{gradshteyn2014table}, Eq. 9.240). Clearly, the probability that a device associates to the RF layer is given by 
\begin{equation}
\label{eq:Assocation_probability_RF}
\mathcal{P_{\mathrm{A_{R}}}} = 1- \mathcal{P_{\mathrm{A_{T}}}}.
\end{equation}
\end{lemma}
\begin{proof}
See \textbf{Appendix B}.
\end{proof}
As the future networks will be highly dense, the density of TBSs can be very high ($\lambda_T \rightarrow \infty$). Also, for indoor applications \cite{chaccour2019reliability}, the absorption loss can approach zero ($k_a(f) \rightarrow 0$). \textcolor{black}{By demonstrating the limit of $\infty$, we mean that the intensity can be quite large but may not be close to infinity.} In both special cases, the association probability can be simplified as.
\begin{corollary} When $\lambda_{T} \rightarrow \infty \Longrightarrow z \rightarrow 0$ and $k_{a} \rightarrow 0 \Longrightarrow z \rightarrow 0$, the argument of   $D_{\nu}\left(z\right)$ in { Lemma~2} tends to zero. For $z = 0$, $D_{\nu}\left(z\right)$ will be simplified to $\frac{\sqrt{\pi}}{2^{\frac{1}{2}b+\frac{1}{4}}\Gamma\left(\frac{3}{4} + \frac{1}{2}b\right)}$, where $b = -\frac{1}{2}- \nu$ and $\Gamma(z)$ is the gamma function (\cite{gradshteyn2014table}, Eq. 8.31). As a result, $\mathcal{P_{\mathrm{A_{T}}}}$ in lemma~2 can be simplified as:
\begin{equation}
\label{eq:Assocation_probability_TWaves2}
\mathcal{P_{\mathrm{A_{T}}}} =\sum_{j=0}^{\infty} \frac{\sqrt{\pi}\left(-1\right)^j\delta_{j}\left(2\beta\right)^{-\frac{v_{j}}{2}}\Gamma\left[v_{j}\right]}{2^{\frac{1}{2}b_{j}+\frac{1}{4}}\Gamma\left(\frac{3}{4} + \frac{1}{2}b_{j}\right)j!}\exp{\left(\frac{-\eta_{j}}{8\beta}\right)},
\end{equation}
where $b_{j} = -\frac{1}{2}- \nu_{j}$.
\end{corollary}
Since a typical user can associate with either RF or THz layer, the total coverage probability can be calculated as:
\begin{equation}
\label{eq:Total-Cov}
\mathbb{C} = \mathcal{P_{\mathrm{A_{T}}}}\mathbb{P}_{C_{\mathrm{T}}} + \mathcal{P_{\mathrm{A_{R}}}\mathbb{P}_{C_{\mathrm{R}}}},
\end{equation}
where $\mathcal{P_{\mathrm{A_{T}}}}$ and $\mathcal{P_{\mathrm{A_{R}}}}$ are defined in {\bf Lemma~2}. Also, $\mathbb{P}_{\mathrm{C_T}}$ and $\mathbb{P}_{C_{\mathrm{R}}}$ refer to the coverage probability conditioned that the typical user associates to a given TBS and RF SBS, respectively. Since the TBSs and RF SBSs are distributed according to different PPPs, the distance of a typical user to its serving BS depends on the tier to which the user is associated. Subsequently, the  distribution of the distance of the typical user to its serving BS in $k$-th tier can be given as:  
\begin{lemma}
\label{lm:Distance_distribution}
The distribution of the distance of a typical user if it is tagged to the THz layer and SBS layer can be given, respectively, as follows:
\small\begin{equation*}
\label{eq:Distance_distribution_THz}
f_{\hat{X}_{T}}\left(\hat{x}\right) = \frac{2 \pi \lambda_T \hat{x}}{\mathcal{P_{\mathrm{A_{T}}}}}
\exp \left(-\pi\lambda_{R} (K \hat{x}^2)^{2/\alpha} e^{2 k_{a}\left(f\right)\hat{x}/\alpha} - \pi\lambda_{T}\hat{x}^2\right),
%\sum_{j=1}^{\infty} \frac{(-1)^j}{j!}\delta_{T,j}x^{\nu_j-1}e^{\left(-\beta_{T} x^2 + \eta_{j} x\right)},
\end{equation*}
\normalsize
\begin{equation*}
\label{eq:Distance_distribution_RF}
f_{\hat{X}_{R}}\left(\hat x\right) \approx \frac{2\pi\lambda_{R} \hat{x}}{\mathcal{P_{\mathrm{A_{R}}}}}  \exp\left(-\pi\lambda_{T}\left(\frac{K \hat x^{\alpha}}{\pi}\right)^{\frac{1}{2+\mu}}- \pi\lambda_{R}\hat{x}^2\right),
%\frac{1}{\mathcal{P}_{A_{R}}} - \frac{1}{\mathcal{P}_{A_{R}}}\sum_{j=1}^{\infty} \frac{(-1)^j}{j!}\delta_{R,j}x^{\nu_j-1}e^{\left(-\beta_{R} x^2 + \eta_{j} x\right)},
\end{equation*}
where $\mu$ is a factor defined in the proof.
\end{lemma}
\begin{proof}
See \textbf{Appendix~C}.
\end{proof}

\begin{lemma}
The calculation of $\mathbb{P}_{\mathrm{C_T}}$ can be performed using the Gil-Pelaez inversion theorem as given in \eqref{eq:Twave_cov_only} where 
\small
\begin{equation*}
\label{eq:CF-total}
\phi_{\Omega}\left(\omega\right) = \mathbb{E}_{\hat{X}_T}\left[\phi_{\Omega \mid \hat{X}_T}\left(\omega\right)\right] = \\ \mathbb{E}_{\hat{X}_T}\left[e^{-j\omega S(\hat x)}\mathcal{L}_{I^T_{\mathrm{agg}}\mid \hat{X}_T}\left(-j\omega\tau_T\right)\right],
\end{equation*}
\normalsize
and the PDF of $\hat{X}_{T}$ is given in \textbf{Lemma \ref{lm:Distance_distribution}}. Considering the interference limited regime, the conditional coverage probability of the user when associated to RF layer can be given as:
\begin{equation}
\label{eq:CondistionRF-cove}
\begin{split}
&\mathbb{P}_{C_{\mathrm{R}}} = \mathrm{P{r}}\left[\chi_0 >\tau_{R}P^{-1}_{R}\gamma^{-1}_{R} \hat x^\alpha I^R_{agg}\right],
\\ & = \mathbb{E}_{I^R_{agg},\hat x}\left[\exp( {-\tau_{R}P^{-1}_{R}\gamma^{-1}_{R} \hat x^\alpha I^R_{agg}})\right],
\\ & = \int_{0}^{\infty}
 \mathcal{L}_{I^R_{agg}}(\tau_{R}P^{-1}_{R}\gamma^{-1}_{R}\hat{x}^\alpha) f_{\hat{X}_{R}}\left(\hat x\right) d\hat{x},
\end{split}
\end{equation}
%where ${L}_{I^R_{agg}}$ is the LT of aggregate interference.
% \begin{equation}
% \label{eq:RF-LT}
% \begin{split}
% &{L}_{I^R_{agg}}\left(s\right) = \exp\left(-2\pi\lambda_{R}\int^{\infty}_{\hat{x}}\left(\frac{1}{1+\left(sP_{R}\gamma_{R}\right)^{-1}x^\alpha}\right)xdx\right).
% \end{split}
% \end{equation}
% Then, averaging over the conditional coverage probability results in the coverage probability of the RF layer; that is:
% \begin{equation}
% \label{eq:RF_Cov_Probab1}
% \begin{split}
%  & \mathbb{P}_{C_{\mathrm{R}}} = \int_{0}^{\infty}
%  \mathcal{L}_{I^R_{agg}}(\tau_{R}P^{-1}_{R}\gamma^{-1}_{R}\hat{x}^\alpha) f_{\hat{X}_{R}}\left(\hat x\right) d\hat{x},
%  %\hat{x}e^{-\pi\lambda_{R}\hat{x}^2 \left(1+\mathcal{Z}\left(\tau_R,\alpha\right)\right) - \pi\lambda_{T}\left(K\hat{x}^{\alpha}\right)^{2\mu}} d\hat{x}
% \end{split}
% \end{equation}
%  From \ref{eq:RF-LT} we have
% \begin{equation}
% \label{eq:RF-LT-2}
% \begin{split}
% &{L}_{I^R_{agg}}\left(\tau_{R}P^{-1}_{R}\gamma^{-1}_{R}\hat{x}^\alpha\right) = \exp\left(-2\pi\lambda_{R}\int^{\infty}_{\hat{x}}\left(\frac{\tau_{R}}{\tau_{R}+\left(\frac{x}{\hat{x}}\right)^\alpha}\right)xdx\right).
% \end{split}
% \end{equation}
% By the use of variable change $u = \left(\frac{x}{\hat{x}}\right)^2\tau_{R}^{-\frac{2}{\alpha}}$ the conditional coverage probability can be written as:
where
% \begin{equation}
% \label{eq:RF-LT-3}
% \begin{split}
$\mathcal{L}_{I^R_{agg}}\left(\tau_{R}P^{-1}_{R}\gamma^{-1}_{R}\hat{x}^\alpha\right) = \exp\left(-\pi\hat{x}^2\lambda_{R}\mathcal{Z}\left(\tau_R,\alpha\right)\right)$,
% \end{split}
% \end{equation}
%where 
% \begin{equation}
% \label{eq:zz}
% \begin{split}
% &\mathcal{Z}\left(\tau_R,\alpha\right) = \tau_{R}^{\frac{2}{\alpha}} \int^{\infty}_{\tau_{R}^{-\frac{2}{\alpha}}}\frac{1}{1+u^\frac{\alpha}{2}}udu
% \end{split}
% \end{equation}
% Given the LT of the aggregate interference in RF layer from \cite{jo2012heterogeneous},  $\mathcal{L}_{I^R_{agg}}(s)=\exp(-\pi \lambda_R  \mathcal{Z}\left(\tau_R,\alpha \right) \hat{x}^2)$, the conditional coverage probability can  be calculated as follows:
% \begin{equation}
% \label{eq:RF_Cov_Probab1}
% \begin{split}
%  & \mathbb{P}_{C_{\mathrm{R}}} = \frac{2\pi\lambda_{R} }{\mathcal{P_{\mathrm{A_{R}}}}}\int_{0}^{\infty}
%  \mathcal{L}_{I^R_{agg}}(s) f_{\hat{X}_{R}}\left(\hat x\right) d\hat{x}
%  %\hat{x}e^{-\pi\lambda_{R}\hat{x}^2 \left(1+\mathcal{Z}\left(\tau_R,\alpha\right)\right) - \pi\lambda_{T}\left(K\hat{x}^{\alpha}\right)^{2\mu}} d\hat{x}
% \end{split}
% \end{equation}
% \begin{figure*}[t]
% \begin{equation}
% \label{eq:RF_Cov_Probab}
%  \begin{split}
%  &\mathbb{P}_{C_{\mathrm{R}}} = \frac{1}{\mathcal{P}_{A_{R}}}  \exp{\left(-\frac{\gamma_{th}\sigma}{P_{R}\gamma_{R}}x^{\alpha} - \pi\lambda r^2\mathcal{Z}\left(\gamma_{th},\alpha\right)\right)} dx  -
%  %\\ &
% %\frac{1}{\mathcal{P}_{A_{R}}} \sum_{j=1}^{\infty} \frac{(-1)^j}{j!} \delta_{R,j} \int_{0}^{\infty}x^{\nu_j-1} \exp{\left(-\frac{\gamma_{th}\sigma}{P_{R}\gamma_{R}}x^{\alpha} - \left(\pi\lambda\mathcal{Z}\left(\gamma_{th},\alpha\right)-\beta_{R}\right)x^2 + \eta_{j} x\right)} dx
% \end{split}
% \end{equation}  
% \end{figure*}
$\mathcal{Z}\left(\tau_R,\alpha\right) = \frac{2\tau_R }{\alpha-2}{}_2F_{1}[1,1-\frac{2}{\alpha};2-\frac{2}{\alpha};-{\tau_R}]$, and ${}_2F_{1}[\cdot]$ is Gauss Hypergeometric function \cite{gradshteyn2014table}. 
\end{lemma}

\begin{remark} 
The  coverage probability of hybrid RF/THz scheme can be given as 
$
\mathbb{C}_{\mathrm{Hybrid}} = 1 - \left(1-\mathcal{P}_{T}\right)\left(1-\mathcal{P}_{R}\right),
$
where $\mathcal{P}_{T}$ is given in Eq. (4), $\mathcal{P}_{R}=2 \pi \lambda_R\int_{0}^{\infty}\rho
 \mathcal{L}_{I^R_{agg}}(\tau_{R}P^{-1}_{R}\gamma^{-1}_{R}\rho^\alpha) \exp(-\pi \lambda_R \rho^2) d\rho$ is the coverage probability of the RF-only network and $\mathcal{L}_{I^R_{agg}}(\cdot)$ is given in {\bf Lemma~4}.
\end{remark}

\section{Numerical Results and Discussions}
% {In this section, we validate our theoretical results with the Monte-Carlo simulations and analyze the performance of a user in a THz-only and  a coexisting RF/THz network. 
% %considering a variety of critical network  parameters. %including density of TBSs and the molecular absorption coefficient.
% } 
\begin{figure*}
\centering
\begin{tabular}{lccccc}
\includegraphics[width=6cm,height=3.25cm]{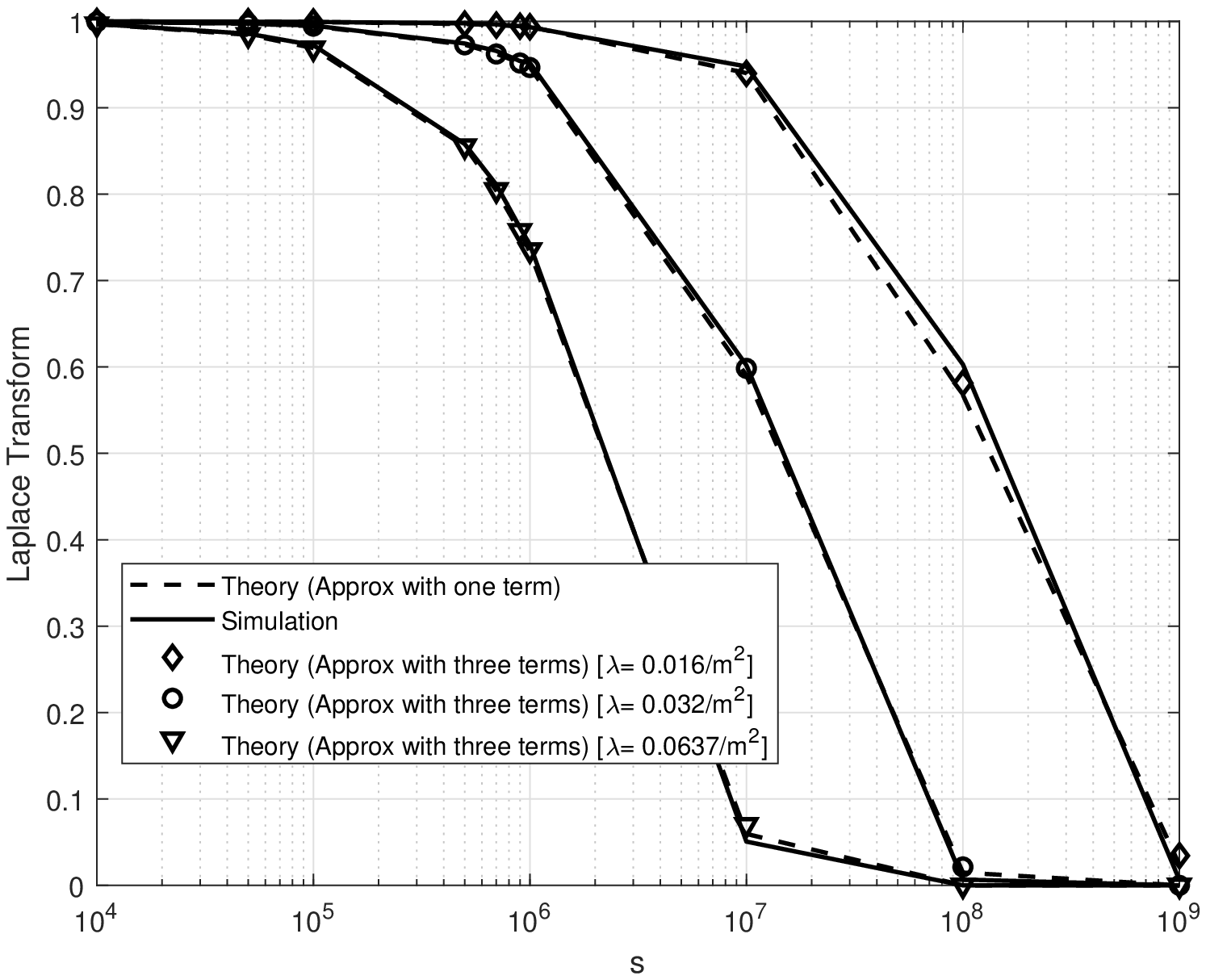}&\hspace{-1cm}
\includegraphics[width=6cm,height=3.25cm]{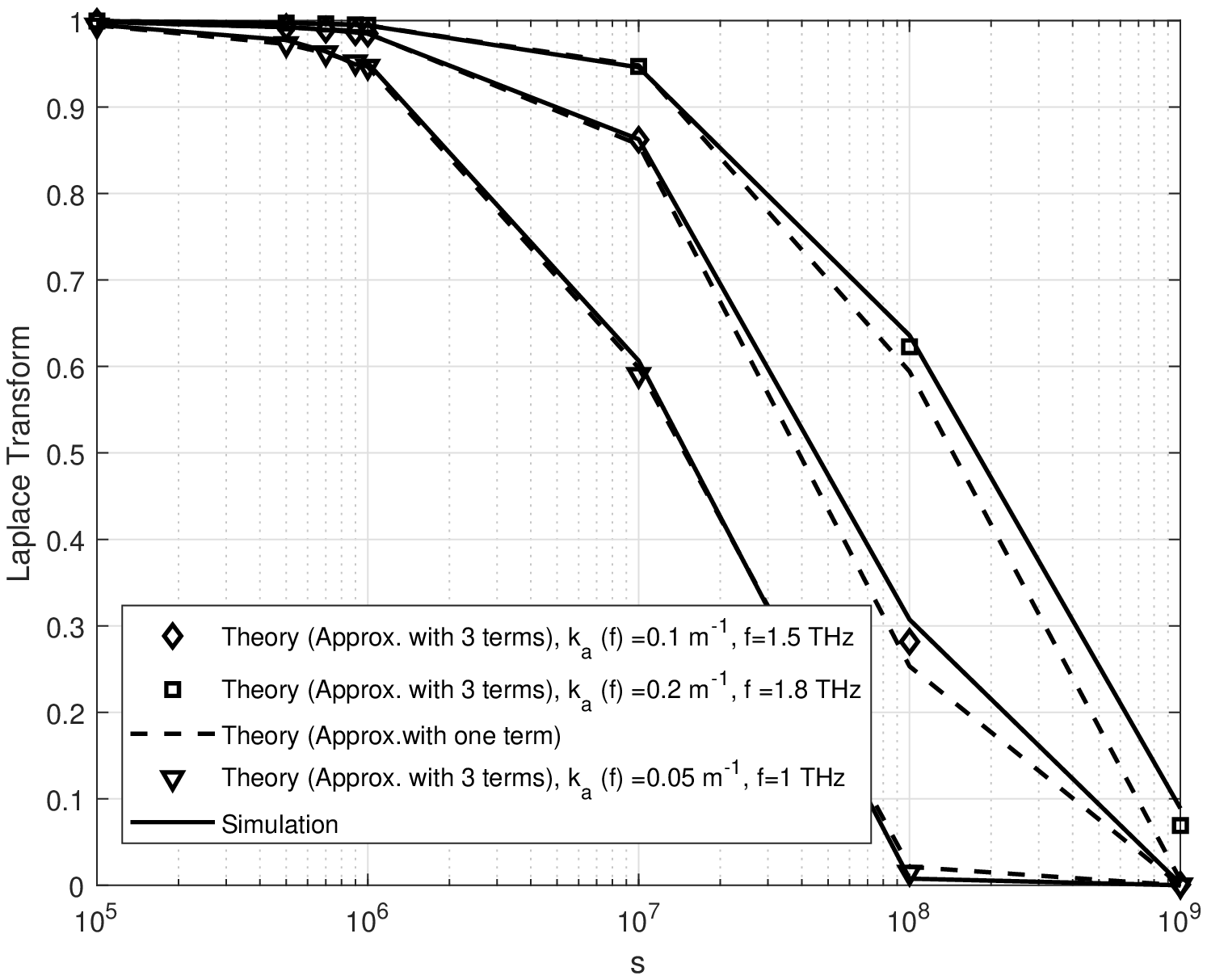}&\hspace{-1cm}
\includegraphics[width=6cm,height=3.25cm]{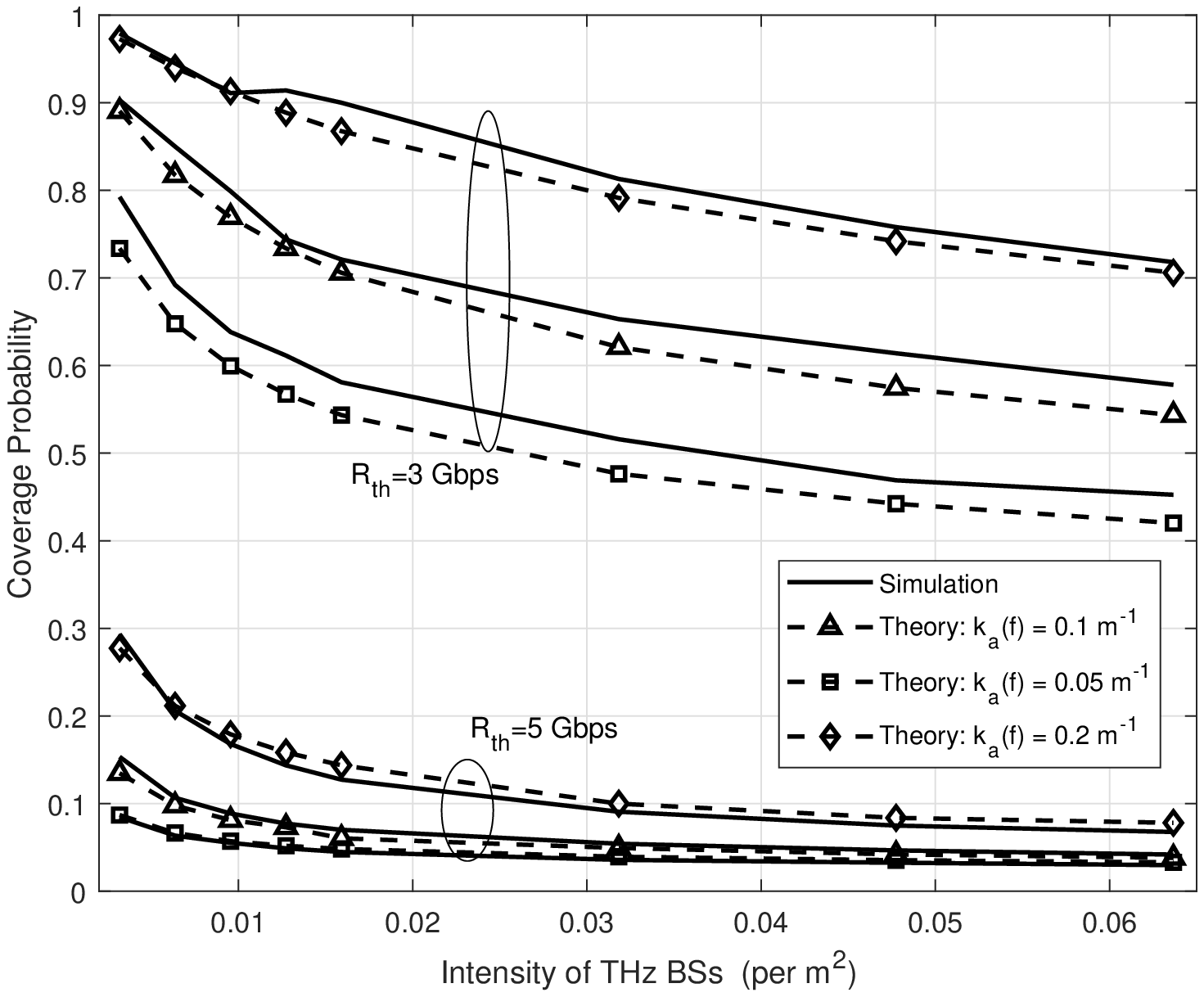}&\\
\qquad \qquad\qquad \qquad (a)  & (b) & (c) 
\end{tabular}
\caption {(a) LT of the aggregate interference as a function of the intensity of TBSs, $k_a(f)=0.05$, $f=$1.0 THz. (b) LT of the aggregate interference as a function of the molecular absorption coefficients, $\lambda_T=0.032$ per m$^2$. (c) Coverage probability of a  user in THz-only network.} 
\label{fig1}
\end{figure*}

\begin{figure*}
\vspace{-5mm}
\centering
\begin{tabular}{lccccc}
\includegraphics[width=6cm,height=3.25cm]{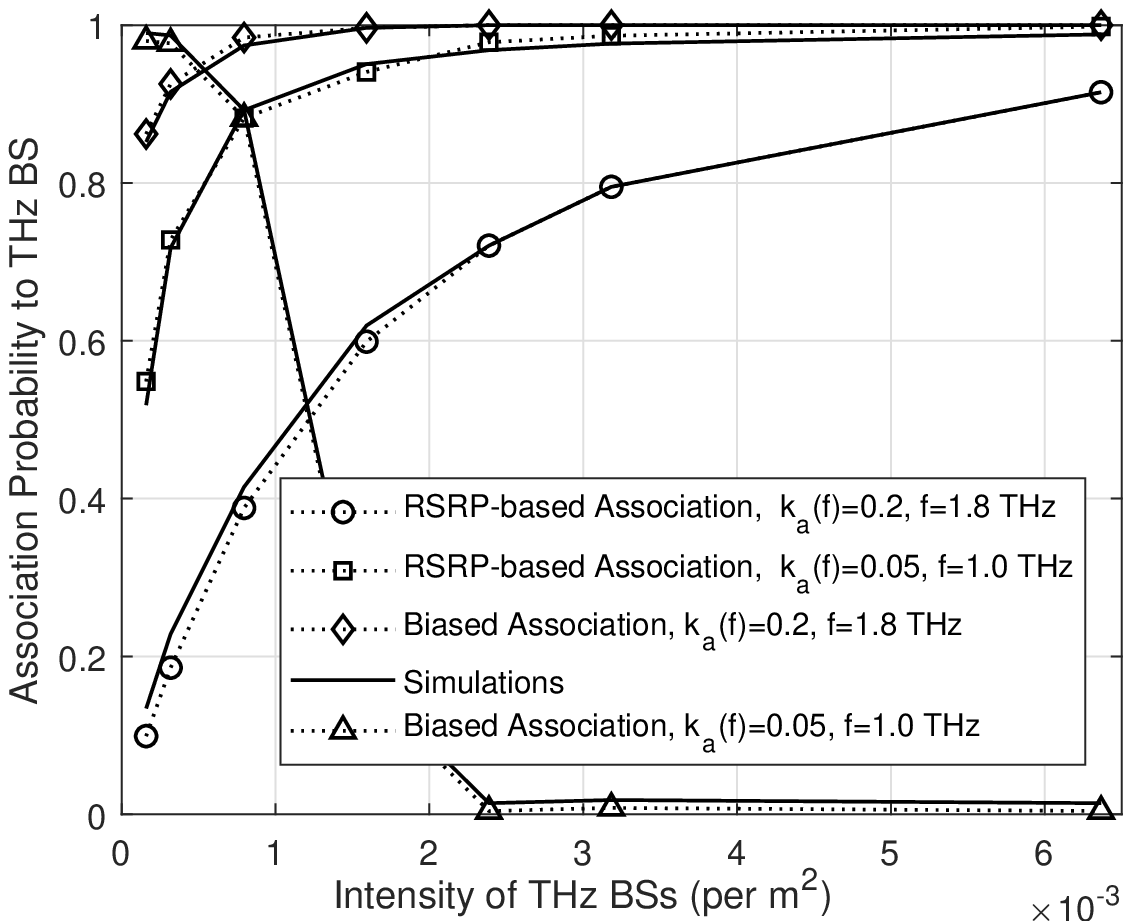}\hspace{-1cm}&\hspace{-1.0cm}
\includegraphics[width=6cm,height=3.25cm]{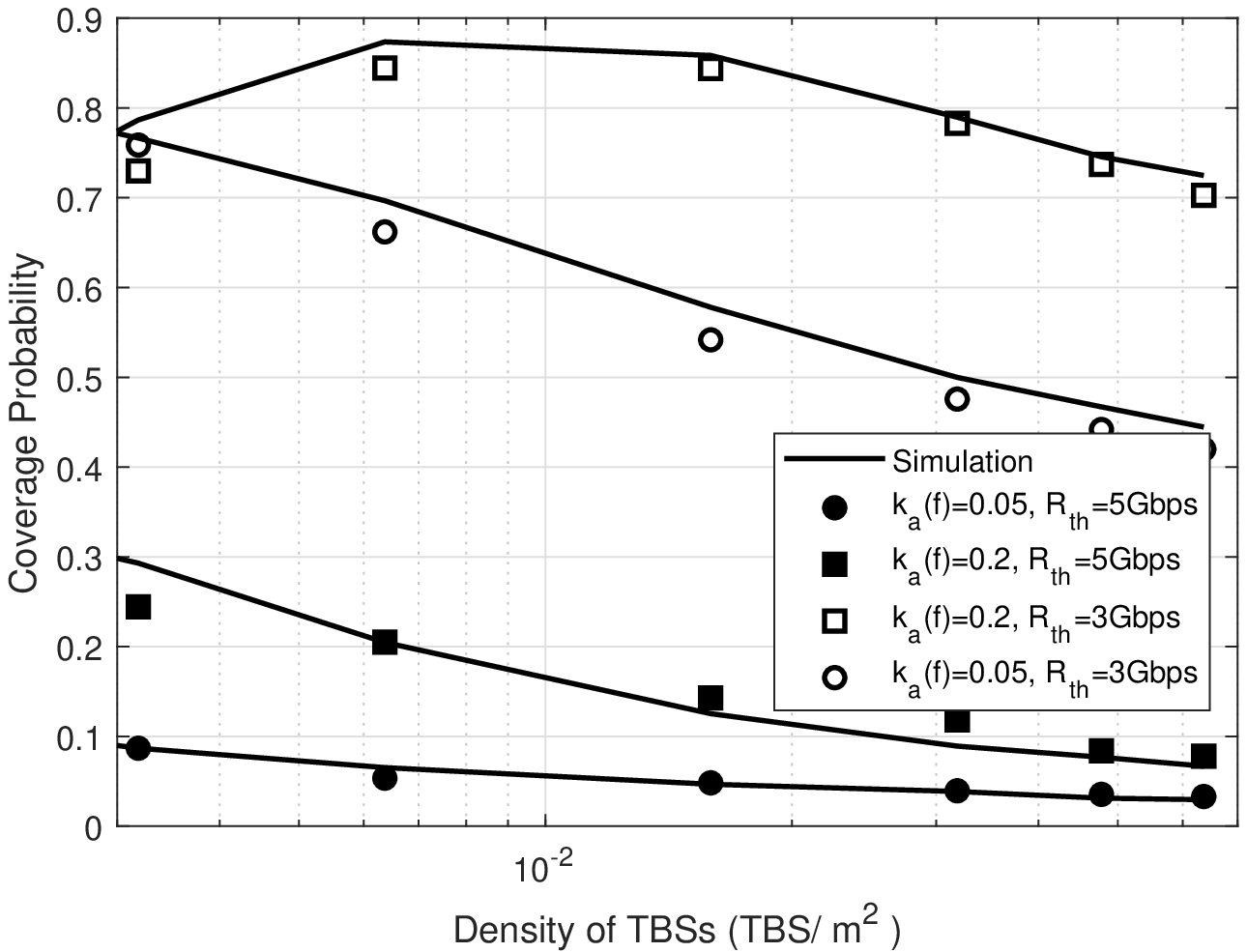}\hspace{-1cm}&\hspace{-01.0cm}
\includegraphics[width=6cm,height=3.25cm]{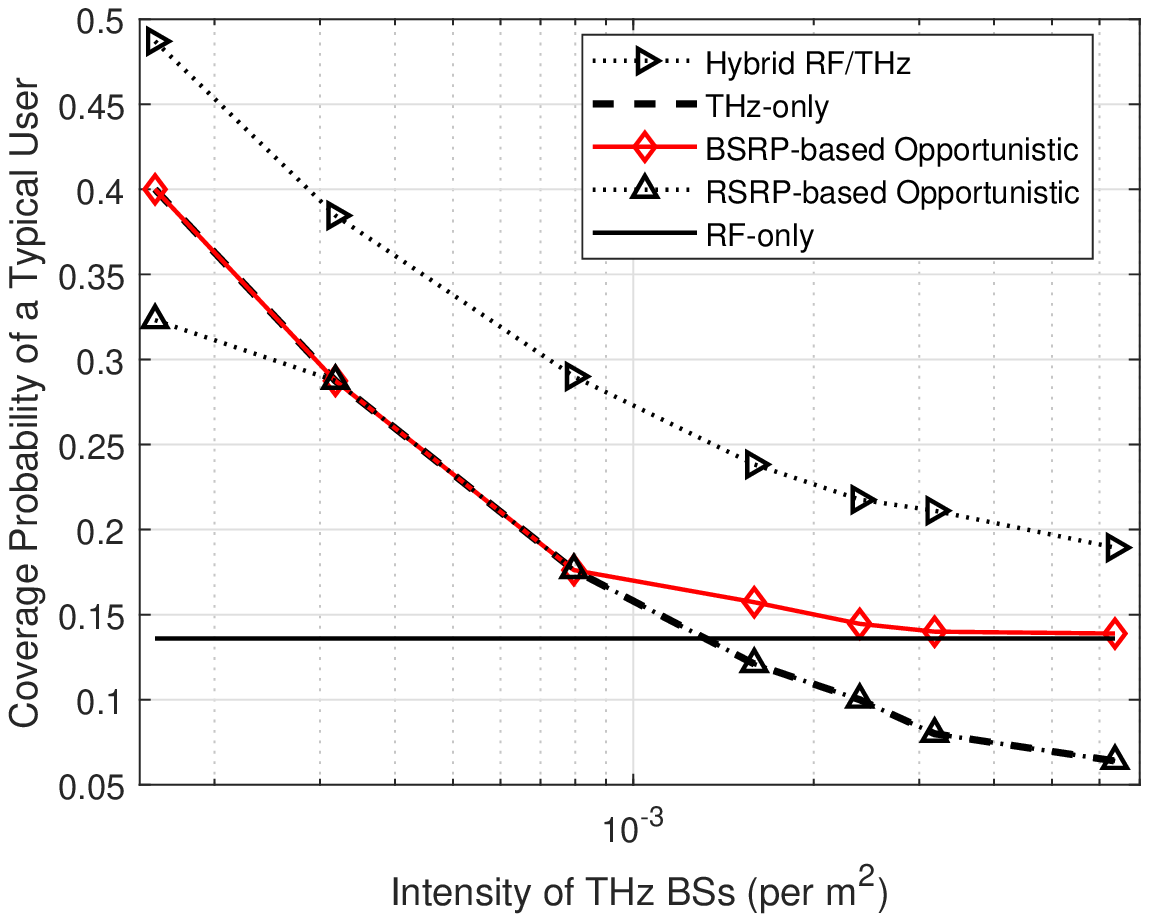}&\\
\qquad \qquad\qquad \qquad (a)  & (b) & (c) 
\end{tabular}
\caption {(a) Association probability as a function of the intensity of TBSs, $k_a(f)=0.2$ $m^{-1}$, $f=$1.8 THz,  $G_{T}^{T}$ and $G_{R}^{T}$ = 15 dB, $\alpha$=3.6, $B$=[1000, 100, 1, 0.001, 0.0001, 0.0001, 0.00001] when $k_a(f)=0.05$m$^{-1}$, $B$=[$10^6,  10^5, 10^4,  10^4, 10^3, 10^3, 10^3]$ when $k_a(f)=0.2$m$^{-1}$. (b) Coverage probability in coexisting network, $\lambda_T=0.1$ m$^{-2}$. (c) Coverage probability in coexisting network, $k_a(f)=0.05$m$^{-1}$, $B$=[$10^3,  10^2, 1,  10^{-3}, 10^{-4}, 10^{-4}, 10^{-5}]$.} 
\label{fig2}
\end{figure*}

Unless stated otherwise, the simulation parameters are listed herein. Users are distributed within a circular disc of radius $100$ m. The antenna gains of TBSs $G_{\mathrm{tx}}^T$ and $G_{\mathrm{rx}}^T$ are set as 25 dB. The transmit powers of TBSs and RF SBSs are   1 W. Three values for  $k_{a}\left(f\right)$ are considered, i.e., 0.05, 0.1, and 0.2 m$^{-1}$ with 1\% of water vapor molecules. These absorption values are chosen from the realistic database and their corresponding central frequencies are $1.0$ THz, $1.5$ THz, and $1.8$ THz, respectively \cite{jornet2011channel}  \cite{rothman2009hitran}. The desired rate threshold is taken as 5 Gbps. The RF transmission frequency is set as $2.1$ GHz and  $\alpha$=2.5. The RF and THz transmission bandwidths are set as 40 MHz and 0.5 GHz, respectively. The intensity of RF SBS $\lambda_R$ is set as 0.0001 BSs/m$^2$.

\textcolor{black}{Fig.~\ref{fig1}(a) depicts the LT of the aggregate interference at the typical user (averaged over large number of realizations of the desired link distance $r$). 
% As the LT of aggregate interference can be used to compute all moments of the aggregate interference, Fig.~\ref{fig1}(a) and Fig.~\ref{fig1}(b) are shown to  evaluate the accuracy of the LT of $I^{\rm T}_{\rm agg}$. 
The theoretic results are calculated by taking the first three terms as well as only one term  of the summation in \textbf{Lemma~1}. There is a close match between the theory and simulations. 
For a given $s$, increasing LT values mean that the aggregate interference is reducing and vice versa.  
In Fig.~1(a), for a given $s$, we note that increasing the intensity of TBSs,  LT decreases rapidly (which implies interference increases). Similarly, for a given $s$, we note that increasing the $k_{a}(f)$ in Fig.~1(b), LT increases (which implies interference decreases due to lower absorption loss at the interfering links). } 
%   We note that increasing the intensity of TBSs, average LT decreases rapidly. Evidently, this is due to the increasing aggregate interference from the THz BSs. Similarly, Fig.~\ref{fig1}(b) shows $\mathbb{E}_r[\mathcal{L}_{I_{\rm agg}}(s)]$ as a function of the molecular absorption coefficient $k_{a}(f)$. 
% There is a close match between the theory (approximation with three terms and one term in \textbf{Lemma~1}) and simulations. 
% It can be seen that decreasing $k_{a}(f)$ decreases the average LT (i.e., increases the aggregate interference). This reduction is due to lower absorption loss at the interfering links. 
The coverage probability of a user in a THz-only network is demonstrated in  Fig.~\ref{fig1}(c). The theoretical results (with first three terms of the infinite summation) show a close match with the simulations. 
We note  that by increasing the molecular absorption coefficient, the coverage probability increases. This is in agreement with in Fig. 1(b).

The association probability to the THz layer is depicted in Fig.~\ref{fig2}(a). For the RSRP-based association,  the probability of user association to the THz layer increases with $\lambda_T$. However, in BRSP, the bias factor $B_T$  is obtained numerically to maximize the coverage probability. 
% When $k_{a}(f) = 0.2$~m$^{-1}$, the interference in THz network is low so the optimal bias values are high, meaning that the higher bias pushes users towards TBSs. Nonetheless, when $k_{a}(f) = 0.05$ m$^{-1}$, the interference is high in THz and therefore the optimal bias value continues to decrease with the increase in TBSs and in turn  pushes more users towards RF SBSs. 
\textcolor{black}{We note that optimal bias to TBSs  decreases with the increase in $\lambda_T$ (which implies increased THz interference). Also,  bias reduction is steep for low $k_a(f)$ (implying higher interference), whereas the reduction in bias is gradual for high values of $k_a(f)$ (implying lower interference).}
Fig.~\ref{fig2}(b) shows the coverage probability of opportunistic RF/THz system with unit bias value.  From the Fig.~\ref{fig2}(a), it is clear that the typical user is likely to associate to the THz layer when bias is unity. Therefore, the total coverage probability is dominated by the behaviour of THz layer and shows the similar behaviour as THz-only network. 
% That is, the coverage probability increases with $k_{a}(f)$  and reduces as the density of TBSs becomes greater. 

Fig.~\ref{fig2}(c) compares the performance of the typical user in coexisting RF/THz network with  THz-only, RF-only, and hybrid RF/THz networks. Hybrid RF/THz network outperforms all networks, since the typical user simultaneously uses both THz and RF transmissions. That is, the additional coverage is at the expense of increased network resources. Coexisting RF/THz with BRSP-based association maximizes the coverage by dynamically adapting to the best tier (since the bias factors are chosen to maximize the coverage) and outperforms RF-only and THz-only schemes. 
%The THz-only scheme follows the performance provided by biased opportunistic RF/THz to some point, but later falls below RF-only network. This is due to increasing interference caused by the large density of TBSs, and biased opportunistic RF/THz avoids the performance impairment by adhering to only the RF layer. 
%Finally, we note that the BRSP-based association in coexisting RF/THz network outperforms . 
%Finally, we note that the RSRP-based association in almost all scenarios. 
%However, it follows the same behaviour as THz only for the large values of the density of TBSs and it even does not perform as satisfactory as RF only scheme, the same reason for performance downturn of quality for THz only plan is true here.

%%%%%%%%%%%%%%%%%%%%%%%%%%%%%%%%%%%%%%%%%%%%%%%%%%%%%%%%%%%%%%%%%%%%%%%%%%%%%%
\section{Conclusion}

We presented a unified stochastic geometry framework  to characterize the performance of  a user in  a coexisting RF/THz network.
% % %Unlike previous works, a semi-closed form yet exact solution provided for LT of aggregate interference, association probability, and the distance distribution of a typical user to its tagged BS. The analytical results are verified through simulations. 
% Hybrid  configuration outperforms all schemes at the expense of extra resource consumption. 
% The opportunistic with BRSP  outperforms  RF-only, THz-only, and opportunistic with the conventional RSRP-based association. 
\textcolor{black}{This work can be extended to incorporate fading by rederiving the LT of the aggregate interference with fading channel statistics.} For blockages, we can follow the approach in \cite{bai2014analysis}. A Boolean blockage model can be considered where the number of blockages in a link are Poisson distributed. Then, LOS probability $ e^{\left[-\left(\xi r + p\right)\right]}$ (where $\xi$ and $p$ are constants)  can be  multiplied  with $\Phi_\Omega|r (\omega)$.

% \textcolor{black}{This work can be extended \textcolor{black}{to include fading and blockages. For the first one we need to by rederive the LT of the aggregate interference}, and for the latter one we follow the approach in \cite{bai2014analysis}.}

%There are some optimization issues in this regard yet to be answered for this system model. For example, the load-aware association between RF and THz layer and beam alignment between a typical user and its tagged TBS are among the problems that needed to be considered in the future works.   

\begin{appendices}
\renewcommand{\theequation}{A.\arabic{equation}}
\setcounter{equation}{0}
\renewcommand{\thesectiondis}[2]{\Alph{section}:}
\section{Proof of Lemma \ref{lm:Laplace_Trans}}
Recall that $I^T_{\mathrm{agg}} =  \sum_{\mathrm{i} \in \Phi_{T} \backslash 0}P_{\mathrm{T}}D_{i} h\left(r_{i}\right)$, after averaging over $D_{i}$ the LT of the aggregate interference can be given as:
\small
\begin{equation*}
\label{eq:Laplace_Trans-Proof}
\begin{split}
&\mathcal{L}_{I^{\mathrm{T}}_{\mathrm{agg}}}(s) = \mathbb{E}_{\Phi_{T}}\left[e^{-sI_{\mathrm{agg}}}\right] 
= \mathbb{E}_{\Phi_{T}}\Biggl[e^{-s F\sum_{i \in \Phi_{T}\backslash 0} P_{\mathrm{T}} \gamma_{T} \frac{e^{-k_{a}\left(f\right)r_i}}{ r_i^2} }\Biggr],
\\&=
\mathbb{E}_{\Phi_{T}}\Biggl[\prod_{i \in \Phi_{T}\backslash 0}\exp\biggl(-s F  P_{\mathrm{T}} \gamma_{T} \frac{e^{-k_{a}\left(f\right)r_i}}{ r_i^2}\Biggr)\Biggr],
\\& 
\stackrel{(a)}{=} \exp\Biggl(-2\pi\lambda_{T}\int_{r}^{\infty}r_i\left(1- \exp{(-s \gamma_{T} F P_{\mathrm{T}} \frac{e^{-k_{a}\left(f\right)r_i}}{ r_i^2})}\right)dr_i\Biggr),
\\&\stackrel{(b)}{=}
\exp\Biggl(-2\pi\lambda_{T}\int_{r}^{\infty}\sum_{l=1}^{\infty}\frac{\left(-s \gamma_{T} F P_{\mathrm{T}} \right)^l \exp\left(-l k_{a}\left(f\right)r_i\right)}{r_i^{2l-1} l!} 
dr_i\Biggr), \\ 
% &\stackrel{(c)}{=}\exp\Biggl(-2\pi\lambda_{T}\sum_{i=1}^{\infty}\frac{\left(-s \gamma_{T} P_{\mathrm{T}}\right)^{i} }{i!}
%  \Biggl[-\left(ik_{a}\left(f\right)\right)^{2\left(i-1\right)}\Gamma\left(-2(i-1),ik_{a}\left(f\right)r)\right)\Biggr]_{r}^{\infty}\Biggr)
%  \\
&\stackrel{(c)}{=}\exp\Biggl(2\pi\lambda_{T}\sum_{l=1}^{\infty}\frac{\left(-s \gamma_{T} F P_{\mathrm{T}}\right)^{l}}{\left(lk_{\alpha}\left(f\right)\right)^{2-2l}l!}
\Gamma\left(2-2l,lk_{a}\left(f\right)r)\right)\Biggr),
\end{split}
\end{equation*}
\normalsize
where (a) is derived by using the probability generating functional (\textbf{PGFL}) with respect to $f(x)=\exp\left(-sP_{T}h\left(r_{i}\right)\right)$, (b) is derived using $\exp\left(-x\right)=\sum_{i=0}^{\infty}(-1)^{i}\frac{x^{i}}{i!}$ (\cite{gradshteyn2014table}, Eq. 1.211), and (c) follows from the {integral identity} $\int \frac{\exp\left(-\beta x^n\right)}{x^m}dx = - \frac{\beta^z\Gamma\left(-z,\beta x^n\right)}{n}$, and $z$ equals to $\frac{m-1}{n}$ (\cite{gradshteyn2014table}, Eq. 2.345). \textcolor{black}{Since the typical user has a distance $r$ from its serving TBS due to the nearest BS association, all interferers exist beyond $r$. Thus, the lower limit in the integral is $r$.} 
   
\section{Proof of Lemma \ref{lm:Assocation_probability}}
\renewcommand{\theequation}{B.\arabic{equation}}
\setcounter{equation}{0}
The distribution of the distances between the typical user and its nearest THz and RF BSs are  $f_{r}(r) = 2\pi \lambda_{T}r\exp\left(-\pi\lambda_{T}r^2\right)$ and $f_{\rho}(\rho) = 2\pi \lambda_{R} \rho \exp\left(-\pi\lambda_{R} \rho^2\right)$, respectively. Thus, averaging over $r$ in (\ref{eq:Asso_Probab_Proof}) yields the association probability\footnote{\textcolor{black}{User association to a BS is a slowly varying process that relies on long-term channel propagation factors such as path-loss and shadowing.}} with TBS as: 
\small
\begin{equation}
\label{eq:Asso_Probab_average}
\begin{split}
&\mathcal{P}_{A_{T}} = \int_{0}^{\infty} 
 \exp\left(-\pi\lambda_{R}\left({K r^2}\right)^{\frac{2}{\alpha}}\exp\left(\frac{2k_{a}\left(f\right)r}{\alpha}\right)\right) f_r(r) dr,
\\ & \stackrel{(a)}{=} \int_{0}^{\infty} 2\pi \lambda_{T}\alpha h^{2\alpha-1} e^{-\pi\lambda_{T}h^{2\alpha}} 
%\\ & \times 
\exp\left(-\pi\lambda_{R}K^{\frac{2}{\alpha}}h^{4} e^{\frac{2k_{a}\left(f\right) h^\alpha}{\alpha}}\right) dh,
\\ & \stackrel{(b)}{=} \int_{0}^{\infty} 2\pi\lambda_{T}\alpha h^{2\alpha-1} e^{-\pi\lambda_{T}h^{2\alpha}} 
%\\ & \times 
\sum_{j=0}^{\infty} \frac{\left(-\pi\lambda_{R}K^{\frac{2}{\alpha}}h^4 e^{\frac{2k_{a}\left(f\right)}{\alpha}h^\alpha}\right)^j}{j!}
%\left(-\pi\lambda_{R}K^{\frac{2}{\alpha}}h^4 e^{\frac{2k_{a}\left(f\right)r}{\alpha}h^\alpha}\right)^j 
dh,
% \\ & = \sum_{j=1}^{\infty} \frac{\left(-1\right)^j}{j!} \int_{0}^{\infty} 2\pi\lambda_{T}\alpha \left(\pi\lambda_{R}\left(\frac{\gamma_{R}}{\gamma_{T}}\right)^{\frac{2}{\alpha}}\right)^j h^{2\alpha-1} h^4  
% \\ & \times \exp\left(-\pi\lambda_{T}h^{2\alpha}\right)\exp\left(\frac{2jk_{a}\left(f\right)h^\alpha}{\alpha}\right) dh
\nonumber\\ & \stackrel{(c)}{=} \sum_{j=0}^{\infty} \frac{ \left(-\pi\lambda_{R}K^{\frac{2}{\alpha}}\right)^j}{j!} \int_{0}^{\infty} 2\pi\lambda_{T} z^{\frac{4j+\alpha}{\alpha}}
%\\ & \times 
e^{-\pi\lambda_{T}z^2 + \frac{2jk_{a}\left(f\right)}{\alpha} z } dz,
\\ & = 
\sum_{j=0}^{\infty} \frac{\left(-\pi\lambda_{R}K^{\frac{2}{\alpha}}\right)^j}{j!} \int_{0}^{\infty} 2\pi\lambda_{T}  z^{v_{j}-1} e^{-\beta z^2 - \eta_{j} z } dz,
% \nonumber\\ & \stackrel{(d)}{=} \sum_{j=1}^{\infty} \frac{\left(-1\right)^j}{j!} \delta_{j}\left(2\beta\right)^{-\frac{v_{j}}{2}}\Gamma\left[v_{j}\right]\exp{\left(\frac{-\eta_{j}}{8\beta}\right)} D_{-v_{j}}\left(\frac{-\eta_{j}}{\sqrt{2\beta}}\right)
\end{split}
\end{equation}
\normalsize
\noindent where (a) is derived by changing variables $r=h^\alpha$, (b) follows from expanding the exponential function as $\exp\left(-x\right)=\sum_{i=0}^{\infty}(-1)^{i}\frac{x^{i}}{i!}$ (\cite{gradshteyn2014table}, Eq. 1.211),  (c) follows from the variable change $z=h^\alpha$, and, finally, \textbf{Lemma~2} is derived by using the integral identity $\int_{0}^{\infty} x^{\nu-1}e^{-\beta x^2-\eta x} dx = \left(2\beta\right)^{-\frac{\nu}{2}}\Gamma\left[\nu\right]\exp\left(\frac{\eta^2}{8\beta}\right)D_{-\nu}\left(\frac{\eta}{\sqrt{2\beta}}\right)$ (\cite{gradshteyn2014table}, Eq. 3.462).
% and taking $\beta = \pi\lambda_{T}$, $\delta_{j} = 2\pi\lambda_{T}\left(\pi\lambda_{R}(\frac{B_{T}P_{R}\gamma_{R}}{P_{T}\gamma_{T}})^{\frac{2}{\alpha}}\right)^j$,  $\nu_{j} = \frac{4j+}{\alpha}+2$, and $\eta_{j} = - \frac{2jk(f)}{\alpha}$. 

\section{Proof of Lemma~3}
\renewcommand{\theequation}{C.\arabic{equation}}
\setcounter{equation}{0}
The distribution of the distance from the tagged BS in the tier $k$ where $k = \{{\rm THz, RF}\}$ can be derived as follows:
\begin{equation}
\label{eq:Distance_distribution_Proof}
\begin{split}
f_{\hat{X}_{k}}\left(\hat x\right) =  &\frac{d\mathrm{P{r}}[\hat{X}_k> \hat x]}{d\hat x} 
= \frac{d\mathrm{P{r}}[X_k>\hat x \vert k = n]}{d\hat x} \\=  &  \frac{d\mathrm{P{r}}[{X}_k>\hat x, k = n]}{\mathrm{P{r}}[k = n]d\hat x},
\end{split}
\end{equation}
where $n \in \{{\rm THz, RF}\}$ is the index of the layer to which a user will associate.  
${\mathrm{P{r}}[k = n]}$ is the association probability of a user to tier $k$ as given in \textbf{Lemma~\ref{lm:Assocation_probability}}. When the user associates to the TBS, the numerator in \eqref{eq:Distance_distribution_Proof} can be given as:
\small
\begin{equation}
\label{eq:Distance_distribution_Proof2}
\begin{split}
&\mathrm{P{r}}[ X_T >\hat x\vert k = {\rm THz}] = \mathrm{P{r}}[ X_T > \hat x,  B_{T}P_r^{\mathrm{THz}} > P_r^{\mathrm{RF}}], \\& =   \int_{\hat x}^{\infty}\mathrm{P{r}}[B_{T}P_r^{\mathrm{THz}} > P_r^{\mathrm{RF}}] f_{X_T}(x) dx, \\& \stackrel{(a)}{=} \int_{\hat x}^{\infty}2\pi \lambda_{T}x e^{ -\pi\lambda_{R} K^{2/\alpha} x^{4/\alpha}e^{2 k_{a}\left(f\right)x/\alpha } -  \pi\lambda_{T}x^2}dx,
\end{split}
\end{equation}
\normalsize
where (a) is derived by substituting $\mathrm{P{r}}[P_r^{\mathrm{THz}} > P_r^{\mathrm{RF}}]$ provided in \textbf{Appendix B}, and $f_{X_T}(x) = 2\pi \lambda_{T}x\exp\left(-\pi\lambda_{T}x^2\right)$. Now substituting \eqref{eq:Distance_distribution_Proof2} in \eqref{eq:Distance_distribution_Proof} results in $f_{\hat{X}_{T}}\left(\hat x\right)$. Likewise, when the user associates to the RF layer, we have: 
\small
\begin{equation}
\label{eq:Distance_distribution_Proof3}
\begin{split}
 & \mathrm{P{r}}[ X_R > \hat x\vert k = {\rm RF}] = \mathrm{P{r}}[ X_R>\hat x, P_r^{\mathrm{RF}} > B_{T}P_r^{\mathrm{THz}}],\\& = \int_{\hat x}^{\infty}\mathrm{P{r}}[P_r^{\mathrm{RF}} > B_{T} P_r^{\mathrm{THz}}] f_{X_{R}}(x) dx,
 %\\& = \int_{\hat x}^{\infty}
 %\mathrm{Pr}\left[P_{{R}}\gamma_{R}\rho^{-\alpha} > P_{{T}}B_{T}\gamma_{T}\frac{\exp\left(-k_{a}\left(f\right)x\right)}{\pi x^2}\right] f_{X_{R}}(x) dx 
 \\& = \int_{\hat x}^{\infty}\mathrm{Pr}\left[\pi r^{2} \exp\left(k_{a}\left(f\right)r\right) > K x^{\alpha} \right] f_{X_{R}}(x) dx,
 \\& \stackrel{(a)}{\approx} \int_{\hat x}^{\infty}\mathrm{Pr}\left[r > \left(\frac{K x^{\alpha}}{\pi}\right)^{\frac{1}{2+\mu}}\right] f_{X_{R}}(x) dx, \\& = \int_{\hat x}^{\infty} 2\pi\lambda_{R} x \exp\left(-\pi\lambda_{T}\left(\frac{K x^{\alpha}}{\pi}\right)^{\frac{1}{2+\mu}} - \pi\lambda_{R}x^2\right) dx,
\end{split}
\end{equation}
where (a) is derived by approximating $r^2\exp\left(k_{a}\left(f\right)r\right)$ with $r^{2 + \mu}$, and $\mu$ is a correcting factor. That is, when $k_{a}\left(f\right) > 0.1$ than  $\mu=2+\frac{10k_{a}\left(f\right)}{1+2k_{a}\left(f\right)}$, otherwise $\mu=2+\frac{15k_{a}\left(f\right)}{1+10k_{a}\left(f\right)}$. Finally, substituting \eqref{eq:Distance_distribution_Proof3} in \eqref{eq:Distance_distribution_Proof} results $f_{\hat{X}_{R}}\left(\hat x\right)$. 
\end{appendices}
\bibliographystyle{IEEEtran}
\bibliography{References}
\end{document}